\documentclass[journal]{IEEEtran}

\usepackage{amsmath,amssymb,amsfonts}
\usepackage{newtxtext,newtxmath}
\usepackage{bm}
\usepackage{cite}
\usepackage{graphicx}
\usepackage[caption=false,font=footnotesize]{subfig}
\usepackage{stfloats}
\usepackage{xcolor}
\colorlet{red}{black}
\definecolor{revisionred}{rgb}{1,0,0}
\newcommand{\revision}[1]{#1}
\usepackage[normalem]{ulem}

\newcommand{\E}{\mathbb{E}}
\newcommand{\tr}{\operatorname{tr}}
\newcommand{\vecop}{\operatorname{vec}}
\newcommand{\CN}{\mathcal{CN}}
\newtheorem{proposition}{Proposition}

\begin{document}

\title{An Operational Leakage Metric for Eavesdropping in Joint Sensing and Communications Systems}

\author{Nanchi~Su,~\IEEEmembership{Member,~IEEE,} Xiaoye Jing~\IEEEmembership{Member,~IEEE,} Fan~Liu,~\IEEEmembership{Senior Member,~IEEE,} George~C.~Alexandropoulos,~\IEEEmembership{Senior Member,~IEEE,} Christos~Masouros,~\IEEEmembership{Fellow,~IEEE,} \\ and Qinyu~Zhang,~\IEEEmembership{Senior Member,~IEEE}
\thanks{N. Su and Q. Zhang are with the Guangdong Provincial Key Laboratory of Aerospace Communication and Networking Technology, Harbin Institute of Technology (Shenzhen), Shenzhen 518055, China. Their email addresses are sunanchi@hit.edu.cn and zqy@hit.edu.cn. Q. Zhang is also with Peng Cheng Laboratory, Shenzhen 518055, China.}
\thanks{X. Jing is with School of Electrical Engineering \& Intelligentization, Dongguan University of Technology, Dongguan 523808, China (e-mail: jingxiaoye@dgut.edu.cn). \textit{Corresponding author: Xiaoye Jing.}}
\thanks{F. Liu is with the National Mobile Communications Research Laboratory, School of Information Science and Engineering, Southeast University, Nanjing 210096, China (e-mail: fan.liu@seu.edu.cn and f.liu@ieee.org).}
\thanks{G. C. Alexandropoulos is with the Department of Informatics and Telecommunications, National and Kapodistrian University of Athens, Athens 16122, Greece (e-mail: alexandg@di.uoa.gr).}
\thanks{C. Masouros is with the Department of Electronic and Electrical Engineering, University College London, London WC1E 7JE, U.K. (e-mail: c.masouros@ucl.ac.uk).}}

\maketitle

\begin{abstract}
In upcoming integrated sensing and communications (ISAC) systems, security depends on whether an eavesdropper (Eve) can recover data symbols or infer sensing information from a common observation with a legitimate receiver. However, existing information and estimation measures do not quantify this dual risk under finite observation records. This letter introduces the novel metric of operational leakage, which signifies Eve’s optimal probability, in the Bayesian sense, of reconstructing communication symbols or estimating a sensing parameter within prescribed distortions. By showing that prior adversarial success probability separates the effect of side information from that of the observation, we propose a \textcolor{red}{conditional Laplace approximation for efficiently evaluating the proposed metric and jointly quantifying the eavesdropping risk to the sensing and communication functionalities}. Numerical results showcase that the risk enabled by adversarial observation increases with Eve's signal-to-noise ratio (SNR), conventional metrics do not uniquely determine finite-record adversarial success probability, and that the proposed conditional Gaussian approximation closely follows numerical Bayes evaluation.
\end{abstract}


\section{Introduction}

\IEEEPARstart{T}{he} concept of integrated sensing and communications (ISAC) uses a common waveform and receiver resources to support data delivery and environmental sensing. This coupling of the two operations also creates a joint security exposure: from the same received signal, an eavesdropper (Eve) may reconstruct the communication symbols as well as infer sensing parameters, e.g., a target direction. Accordingly, a secure ISAC design should assess Eve's capability on both tasks under an explicitly specified knowledge model \cite{su2026secureisac}.

Most existing physical layer security (PLS) measures quantify an information rate, an average information exposure, or an estimation bound \cite{wyner1975wiretap,yamamoto1997rate,su2022secure,wei2022multifunctional}. Recent ISAC studies also emphasize security and privacy threats created by the dual functions, including artificial noise (AN)-based protection and sensing-assisted Eve estimation \cite{su2026secureisac,su2024sensingassisted,gavras2025secureisac,bazzi2024securefd}. These measures remain indispensable for analysis and design, but do not directly answer the following important ISAC-oriented question: given a finite observation record, what is the probability that an adversary attains the reconstruction or sensing accuracy specified by the ISAC service? This distinction is particularly important when an attack succeeds as soon as Eve meets either of the prescribed distortion thresholds.

This letter addresses the latter question using a Bayesian decision formulation related to posterior vulnerability and generalized gain functions \cite{smith2009foundations,alvim2012gain}. In particular, we apply tools of leakage theory to ISAC by jointly treating symbol reconstruction and sensing estimation, while explicitly considering what Eve knows before receiving an observation from what its observation actually reveals. The main contributions of this letter are as follows:

\begin{itemize}
\item We define an operational metric as Eve's optimal Bayesian probability of meeting either a communication or sensing distortion threshold.

\item We introduce a baseline based on prior information and a normalized observation gain that separates existing side information from the observation risk.

\item We derive a conditional Laplace approximation based on the posterior Hessian and assess it against numerical Bayes results for feasible designs.
\end{itemize}

The remainder of this letter is organized as follows. Sections II and III describe the observation model and Eve's Bayesian decision rules. Sections IV and V include the proposed operational metric and present its conditional Gaussian evaluation. Section VI reports our numerical results, while Section VII concludes the letter.

\section{System and Adversary Observation Model}
We consider a full-duplex (FD) ISAC base station (BS)~\cite{alexandropoulos2022fdmimo} with $N_{\mathrm{BS},t}$ transmit antennas serving $K$ single-antenna downlink users over $L$ snapshots. \textcolor{red}{The BS receive chain simultaneously acquires monostatic sensing echoes \cite{smida2024fd}. These signals are not used below because the proposed leakage metric concerns Eve's observation.} At snapshot $\ell\in\{1,\dots,L\}$,this FD node transmits the following composite waveform:
\begin{equation}
\mathbf{x}(\ell)=\mathbf{F}\mathbf{s}(\ell)+\mathbf{z}(\ell),
\label{eq:x}
\end{equation}
where $\mathbf{s}(\ell)\in\mathbb{C}^{K\times 1}$ is the data symbol vector, $\mathbf{F}\in\mathbb{C}^{N_{\mathrm{BS},t}\times K}$ is the precoder, and $\mathbf{z}(\ell)\sim\CN(\mathbf{0},\mathbf{Z})$ is AN independent of $\mathbf{s}(\ell)$, with $\mathbf{Z}\succeq\mathbf{0}$ \cite{goel2008an}.
The total transmit power satisfies the condition: $\tr(\mathbf{F}\mathbf{F}^{H}+\mathbf{Z})\le P_{\max}$.

Eve, equipped with $N_E$ antennas, observes
\begin{equation}
\mathbf{y}_E(\ell)=\mathbf{H}_{\mathrm{eff}}(\boldsymbol{\phi},\ell)\,\mathbf{x}(\ell)+\mathbf{n}_E(\ell),
\label{eq:yE}
\end{equation}
where $\mathbf{n}_E(\ell)\sim\CN(\mathbf{0},\sigma_E^2\mathbf{I}_{N_E})$ is independent receiver noise and
$\boldsymbol{\phi}=[\theta_E,\theta_L]^T$ collects sensing-related parameters of interest.
The effective channel $\mathbf{H}_{\mathrm{eff}}(\boldsymbol{\phi},\ell)\in\mathbb{C}^{N_E\times N_{\mathrm{BS},t}}$ captures both the bistatic echo from the BS through the target to Eve and the direct path from the BS to Eve. We assume a narrowband, far-field model with one point reflector (the sensing target) and one reflection from it. The gain $\beta(\ell)$ includes the BS--target--Eve propagation as well as the target reflection, while other reflectors and multipath are not included. This gain is modeled as
\begin{equation}
\mathbf{H}_{\mathrm{eff}}(\boldsymbol{\phi},\ell)
=\beta(\ell)\,\mathbf{b}(\theta_E)\mathbf{a}_t(\theta_L)^{H}+\mathbf{H}_{\mathrm{CE}},
\label{eq:Heff}
\end{equation}
where $\mathbf{b}(\theta_E)$ and $\mathbf{a}_t(\theta_L)$ are receive/transmit steering vectors, and $\mathbf{H}_{\mathrm{CE}}$ is the direct BS-Eve channel matrix. Specifically, $\mathbf{b}(\theta_E)$ is Eve's receive-array steering vector for an echo arriving from angle $\theta_E$, and $\mathbf{a}_t(\theta_L)$ is the BS transmit-array steering vector for a signal sent toward the target at angle $\theta_L$.

For processing across snapshots, we use the stacked matrices $\mathbf{S}\in\mathbb{C}^{K\times L}$ and $\mathbf{Y}_E\in\mathbb{C}^{N_E\times L}$, whose column $\ell$ entries are $\mathbf{s}(\ell)$ and $\mathbf{y}_E(\ell)$, respectively. Unless otherwise stated, snapshots are independent across $\ell$. In addition, we henceforth assume that the symbol matrix $\mathbf{S}$ and sensing parameter $\boldsymbol{\phi}$ are assigned priors induced by Eve's side information. \textcolor{red}{We distinguish Eve's realized side information, denoted by $V_{\mathcal{K}_E}$, from design variables and fixed model hyperparameters, collected in $\boldsymbol{\omega}$. Both are known to Eve before it observes $\mathbf{Y}_E$.} We collect the unknown physical variables in $\boldsymbol{\nu}$. \textcolor{red}{Specifically, $\boldsymbol{\omega}$ is defined as follows:}
\begin{equation}
\boldsymbol{\omega}
\triangleq
\big(\mathbf{F},\mathbf{Z},\sigma_E^2,\boldsymbol{\omega}_{\rm aux}\big)
\label{eq:omega_tuple}
\end{equation}
\textcolor{red}{Here, $\boldsymbol{\omega}_{\rm aux}$ contains fixed auxiliary quantities, such as the array geometry, carrier wavelength, and channel model parameters. All remaining uncertain physical quantities are collected in $\boldsymbol{\nu}$.}

\section{Adversary Side Information and Bayesian Rules}
Given the previously presented observation model, the leakage metric quantifying Eve's adversary capability depends on the following specifications: first, what Eve knowns before observing $\mathbf{Y}_E$ and, second, what Bayesian decision rules Eve may use after observing $\mathbf{Y}_E$.

\subsection{Side Information and Latent Quantities}
We model Eve's knowledge by
\begin{equation}
\mathcal{K}_E \triangleq
\Big(\mathcal{K}_E^{\rm wf},\, \mathcal{K}_E^{\rm ch},\, \mathcal{K}_E^{\rm sync},\, \mathcal{K}_E^{\rm prior}\Big),
\label{eq:KE_tuple}
\end{equation}
where $\mathcal{K}_E^{\rm wf}$ represents waveform information, $\mathcal{K}_E^{\rm ch}$ is channel and array information, $\mathcal{K}_E^{\rm sync}$ is synchronization and calibration information, and $\mathcal{K}_E^{\rm prior}$ denotes any knowledge of prior distributions for the unknown symbol matrix $\mathbf{S}$, sensing parameter $\theta_L$, and physical variables $\boldsymbol{\nu}$. The tuple $\mathcal{K}_E$ specifies the categories of side information, while $V_{\mathcal{K}_E}$ denotes their realized values, such as known pilots, channel estimates, timing offsets, calibration data, and prior parameters available to Eve.

The public variable $\boldsymbol{\omega}$ should not be confused with Eve's realized side information. If $\mathbf{H}_{\mathrm{CE}}$, $\beta(\ell)$, $\theta_E$, synchronization offsets, or calibration parameters are not included in $V_{\mathcal{K}_E}$, they enter $\boldsymbol{\nu}$. Parameters known to Eve are removed from $\boldsymbol{\nu}$. If $\mathbf{F}$, $\mathbf{Z}$, or $\sigma_E^2$ is not public, its realization must also be placed in $\boldsymbol{\nu}$ and only its model hyperparameters remain in $\boldsymbol{\omega}$.

For fixed public variables $\boldsymbol{\omega}$, Eve's Bayesian model uses the joint prior $p_{\mathcal{K}_E}(\mathbf{S},\theta_L,\boldsymbol{\nu}\mid V_{\mathcal{K}_E},\boldsymbol{\omega})$ for the unknown symbols, sensing parameter, and physical variables. \textcolor{red}{Starting from \eqref{eq:yE}, the covariance matrix of Eve's observation after marginalizing the AN is computed as follows:}
\begin{equation}
\begin{aligned}
\mathbf{C}_E(\ell)
&\triangleq
\sigma_E^2\mathbf{I}_{N_E}
+\mathbf{H}_{\rm eff}(\theta_L,\boldsymbol{\nu},\ell)
\mathbf{Z}\mathbf{H}_{\rm eff}^{H}(\theta_L,\boldsymbol{\nu},\ell),
\end{aligned}
\label{eq:CE_cov}
\end{equation}
yielding the following expression for the conditional likelihood marginalized over AN resulting from $L$ independent snapshots:
\begin{equation}
\begin{aligned}
&p_{\boldsymbol{\omega}}\big(\mathbf{Y}_E \mid \mathbf{S},\theta_L,\boldsymbol{\nu},V_{\mathcal{K}_E}\big) \\
&\;\;\;\;\;\;\;\;\;=\prod_{\ell=1}^{L} \CN\!\Big(\mathbf{y}_E(\ell)\,\big|\,
\mathbf{H}_{\rm eff}(\theta_L,\boldsymbol{\nu},\ell)\mathbf{F}\mathbf{s}(\ell), \mathbf{C}_E(\ell)\Big).
\end{aligned}
\label{eq:likelihood_condS}
\end{equation}
Together with the prior on $\mathbf{S}$, $\theta_L$, and $\boldsymbol{\nu}$, \eqref{eq:likelihood_condS} defines Eve's posterior distribution and, hence, the observation-assisted adversarial success probability.
Equation \eqref{eq:likelihood_condS} displays the case where $\mathbf{F}$, $\mathbf{Z}$, and $\sigma_E^2$ are available to Eve. If any of them is unknown, the corresponding realized quantity is included in $\boldsymbol{\nu}$, and the same likelihood is interpreted with that component supplied by $\boldsymbol{\nu}$.

\subsection{Bayesian Decision Rules}
\textcolor{red}{For the exact leakage metric, Eve may use any measurable Bayesian rule based only on its observation, side information, model information known to Eve, and private randomization. We define the following rule class:}
\begin{equation}
\begin{aligned}
\mathcal{A}_{\rm Bayes}(\mathcal{K}_E)\triangleq
\Big\{\, g\ \Big|\ &
g \text{ is measurable with input }
(\mathbf{Y}_E,V_{\mathcal{K}_E},\\
&\boldsymbol{\omega},U),
 U\mathrel{\perp\!\!\!\perp}(\mathbf{Y}_E,\mathbf{S},\theta_L,\boldsymbol{\nu})
\Big\}.
\end{aligned}
\label{eq:capability_class}
\end{equation}
For $g\in\mathcal{A}_{\rm Bayes}(\mathcal{K}_E)$, we have $g(\mathbf{Y}_E,V_{\mathcal{K}_E},\boldsymbol{\omega},U)=(\widehat{\mathbf{S}},\widehat{\theta}_L)$, where $\mathrel{\perp\!\!\!\perp}$ denotes statistical independence and $U$ is Eve's independent private randomization. The class imposes no constraint on computational complexity or estimator architecture, so restricted neural estimators, decoder families, or iterative algorithms define separate constrained leakage metrics.

\section{Operational Joint Leakage Metric}
Eve has two reconstruction objectives: the communication symbols $\mathbf{S}$ and the sensing parameter $\theta_L$. Let $\mathbf{s}=\vecop(\mathbf{S})$ and $n_C=KL$. \revision{Under the Gaussian signaling assumption, we measure symbol reconstruction accuracy by the squared error per complex symbol:}
\begin{equation}
d_c(\widehat{\mathbf{S}},\mathbf{S})
\triangleq
\frac{1}{n_C}\|\widehat{\mathbf{s}}-\mathbf{s}\|_2^2,
\label{eq:dc}
\end{equation}
\revision{This quantifies the reconstruction error in $\mathbf{S}$ per complex symbol. For finite-alphabet data detection, a suitable alternative is the normalized Hamming distortion $d_c^{\rm H}(\widehat{\mathbf{S}},\mathbf{S})=n_C^{-1}\sum_{i=1}^{n_C}\mathbf{1}\{\widehat{s}_i\ne s_i\}$, with estimates restricted to the symbol alphabet. Its zero-threshold success event is exact recovery of the symbol block. The Bayesian definitions below also apply to this discrete loss. The Gaussian evaluation and numerical results here use \eqref{eq:dc}.} For sensing, let $d_s(\widehat{\theta}_L,\theta_L)$ be an error measure chosen for the sensing task. For a target direction, it can be the smallest difference between the estimated and true angles, treating angles separated by $2\pi$ as the same: $\min_{k\in\mathbb{Z}}|\widehat{\theta}_L-\theta_L+2\pi k|$. Given $\delta_c\ge0$ and $\delta_L>0$, we define the sets:
\begin{align}
\mathcal{E}_{\mathrm{c}}(\delta_c)
&\triangleq \{d_c(\widehat{\mathbf{S}},\mathbf{S})\le\delta_c\},
\label{eq:Ec}\\
\mathcal{E}_{\mathrm{s}}(\delta_L)
&\triangleq \{d_s(\widehat{\theta}_L,\theta_L)\le\delta_L\}.
\label{eq:Es}
\end{align}
An attack is successful when Eve satisfies either distortion threshold, which defines the following success event (logical OR):
\begin{equation}
\mathcal{E}_{\rm OR}(\delta_c,\delta_L)
\triangleq
\mathcal{E}_{\mathrm{c}}(\delta_c)\cup\mathcal{E}_{\mathrm{s}}(\delta_L),
\label{eq:Ecomp_OR}
\end{equation}
\textcolor{red}{Note that Eve may use separate symbol and sensing estimates, i.e., the estimates need not come from one joint estimate. In this letter, we focus on success in either task. One may instead require success in both tasks or assign weights to the two tasks, this is left for future work.}

\textcolor{red}{For any probability law $\pi$ of the parameter pair $(\mathbf{S},\theta_L)$, such as Eve's prior $P_{\mathbf{S},\theta_L\mid V_{\mathcal{K}_E},\boldsymbol{\omega}}$ or posterior $P_{\mathbf{S},\theta_L\mid \mathbf{Y}_E,V_{\mathcal{K}_E},\boldsymbol{\omega}}$, define the OR vulnerability as follows:}
\begin{equation}
\begin{aligned}
\mathcal{V}_{\rm OR}(\pi)
\triangleq
\sup_{\widehat{\mathbf{S}},\widehat{\theta}_L}
\pi\big(
&d_c(\widehat{\mathbf{S}},\mathbf{S})\le\delta_c\text{ OR }
d_s(\widehat{\theta}_L,\theta_L)\le\delta_L
\big).
\end{aligned}
\label{eq:vulnerability_OR}
\end{equation}
\textcolor{red}{Let $\pi_{\mathrm{c}}$ and $\pi_{\mathrm{s}}$ be the communication and sensing marginals, and let $\mathcal{V}_{\mathrm{c}}(\pi_{\mathrm{c}})$ and $\mathcal{V}_{\mathrm{s}}(\pi_{\mathrm{s}})$ represent the respective suprema over the single-task events.}

\begin{proposition}
For any law $\pi$ of $(\mathbf{S},\theta_L)$,
\begin{equation}
\begin{aligned}
\max\{\mathcal{V}_{\mathrm{c}}(\pi_{\mathrm{c}}),\mathcal{V}_{\mathrm{s}}(\pi_{\mathrm{s}})\}
&\le \mathcal{V}_{\rm OR}(\pi)\\
&\le \min\{1,\mathcal{V}_{\mathrm{c}}(\pi_{\mathrm{c}})+\mathcal{V}_{\mathrm{s}}(\pi_{\mathrm{s}})\}.
\end{aligned}
\label{eq:or_bounds}
\end{equation}
If $\pi=\pi_{\mathrm{c}}\pi_{\mathrm{s}}$, then
\begin{equation}
\mathcal{V}_{\rm OR}(\pi)
=
1-\big(1-\mathcal{V}_{\mathrm{c}}(\pi_{\mathrm{c}})\big)
\big(1-\mathcal{V}_{\mathrm{s}}(\pi_{\mathrm{s}})\big).
\label{eq:or_independent}
\end{equation}
Moreover, $\mathcal{V}_{\rm OR}$ is convex in $\pi$.
\end{proposition}

\begin{IEEEproof}
For any fixed pair of estimates, let $A$ and $B$ denote the events that the symbol and sensing estimates meet their respective distortion thresholds. The OR success event is $A\cup B$. Hence, $\max\{\Pr(A),\Pr(B)\}\leq\Pr(A\cup B)\leq\min\{1,\Pr(A)+\Pr(B)\}$. If $\pi=\pi_{\mathrm c}\pi_{\mathrm s}$, $\mathbf{S}$ and $\theta_L$ are independent, and $\Pr(A\cup B)=1-[1-\Pr(A)][1-\Pr(B)]$. This expression is increasing in both single-task success probabilities. Therefore, separately maximizing the communication and sensing success probabilities also maximizes the OR success probability. Finally, for fixed estimates, $\Pr(A\cup B)$ is linear in $\pi$. Taking the supremum over all estimates thus yields a convex function of $\pi$.
\end{IEEEproof}

\subsection{Prior Baseline and Observation Assisted Risk}
For compact notation, let $\zeta\triangleq(\delta_c,\delta_L,\mathcal{K}_E,\boldsymbol{\omega})$. The prior adversarial success probability is defined as follows:
\begin{equation}
P_{0}^{\star}(\zeta)
\triangleq
\E_{V_{\mathcal{K}_E}}\!\left[
\mathcal{V}_{\rm OR}\!\left(
P_{\mathbf{S},\theta_L\mid V_{\mathcal{K}_E},\boldsymbol{\omega}}
\right)
\right],
\label{eq:P0_star}
\end{equation}
where the adversary chooses estimates using its side information but not the observation $\mathbf{Y}_E$. The observation-assisted adversarial success probability is
\begin{equation}
P_{Y}^{\star}(\zeta)
\triangleq
\E_{\mathbf{Y}_E,V_{\mathcal{K}_E}}\!\left[
\mathcal{V}_{\rm OR}\!\left(
P_{\mathbf{S},\theta_L\mid \mathbf{Y}_E,V_{\mathcal{K}_E},\boldsymbol{\omega}}
\right)
\right],
\label{eq:posterior_characterization}
\end{equation}
In both expressions, $\mathcal{V}_{\rm OR}$ uses the distribution of $(\mathbf{S},\theta_L)$ after averaging over $\boldsymbol{\nu}$. Equation \eqref{eq:posterior_characterization} gives the highest adversarial success probability over $g\in\mathcal{A}_{\rm Bayes}(\mathcal{K}_E)$ because Eve chooses the estimates that maximize the conditional success probability for each observation. If no rule reaches this highest value exactly, rules can get arbitrarily close to it. These definitions give $P_Y^\star\ge P_0^\star$. More generally, giving Eve a less informative observation cannot increase its adversarial success probability.

The observation-induced leakage provides the increase in adversarial success probability beyond the prior baseline. We use the following normalized form
\begin{equation}
L_{\rm op}^{\rm norm}(\zeta)
\triangleq
\frac{P_Y^{\star}(\zeta)-P_0^{\star}(\zeta)}
{1-P_0^{\star}(\zeta)}
\label{eq:Lop_norm}
\end{equation}
when $P_0^{\star}<1$. It is the observation-induced increase divided by the remaining chance that Eve fails using side information alone, $1-P_0^\star$. Thus, a value of zero implies that the observation provides no additional advantage beyond Eve's side information. If $P_0^{\star}=1$, the normalized value is undefined and one should report $P_0^\star$. Results should report $(P_0^\star,P_Y^\star,L_{\rm op}^{\rm norm})$ together. \textcolor{red}{Both $P_0^{\star}$ and $P_Y^{\star}$ are nondecreasing in each distortion threshold: increasing $\delta_c$ or $\delta_L$ enlarges the corresponding success event and, therefore, cannot decrease either success probability.}

\textcolor{red}{The operational leakage metric can be evaluated under any observation model $p_{\boldsymbol{\omega}}(\mathbf{Y}_E\mid \mathbf{S},\theta_L,\boldsymbol{\nu},V_{\mathcal{K}_E})$. For example, \eqref{eq:Heff} can be replaced by a model incorporating multipath, multiple targets, or wideband signaling without changing the Bayesian definition.}

\section{Conditional Gaussian Evaluation}
The exact posterior characterization in expression in \eqref{eq:posterior_characterization} can hardly have a closed form. This section gives a conditional Laplace approximation that mirrors its decision structure: optimize the posterior success probability for each observation, then average over observations. \revision{This approximation uses the squared-error criterion in \eqref{eq:dc} for the communication task under Gaussian signaling. Finite-alphabet detection instead requires posterior sums over discrete symbol hypotheses, possibly combined with integration over continuous sensing and nuisance parameters. The Gaussian formulas below do not directly evaluate symbol or coded-message decoding success.}

\subsection{Conditional Laplace Posterior}
Let $\boldsymbol{\nu}$ collect the nuisance parameters unknown to Eve. Since $\mathbf{S}$ is complex, we define the following vector:
\begin{equation}
\mathbf{s}_{\rm R}
\triangleq
\big[\Re\{\vecop(\mathbf{S})\}^{T},\
\Im\{\vecop(\mathbf{S})\}^{T}\big]^{T}
\in\mathbb{R}^{2n_C},
\label{eq:s_real}
\end{equation}
and organize the local parameter and target vectors as
\begin{equation}
\boldsymbol{\xi}
\triangleq
\big[\boldsymbol{\nu}^{T},\theta_L,\mathbf{s}_{\rm R}^{T}\big]^{T},
\qquad
\boldsymbol{\tau}
\triangleq
\big[\theta_L,\mathbf{s}_{\rm R}^{T}\big]^{T}.
\label{eq:xi_tau}
\end{equation}
For each realized $(\mathbf{Y}_E,V_{\mathcal{K}_E})$, define
\begin{equation}
\ell_{\rm post}(\boldsymbol{\xi})
\triangleq
\log p_{\boldsymbol{\omega}}(\mathbf{Y}_E|\boldsymbol{\xi},V_{\mathcal{K}_E})
+\log p_{\mathcal{K}_E}(\boldsymbol{\xi}|V_{\mathcal{K}_E},\boldsymbol{\omega}).
\label{eq:ell_post}
\end{equation}
The maximum a posteriori (MAP) point and observed posterior Hessian are respectively defined as
\begin{equation}
\widehat{\boldsymbol{\xi}}_{\rm MAP}
\triangleq
\arg\max_{\boldsymbol{\xi}}\ell_{\rm post}(\boldsymbol{\xi}),
\label{eq:map}
\end{equation}
and
\begin{equation}
\mathbf{H}_{\rm post}(\mathbf{Y}_E,V_{\mathcal{K}_E})
\triangleq
-\nabla_{\boldsymbol{\xi}}^{2}
\ell_{\rm post}(\boldsymbol{\xi})
\big|_{\boldsymbol{\xi}=\widehat{\boldsymbol{\xi}}_{\rm MAP}} .
\label{eq:Hpost}
\end{equation}

% Removed the manual column break so revised text flows without a blank column.
Fisher information and the Cram\`er--Rao bound (CRB) are often used as observation-independent ISAC design criteria \cite{su2025bcrb}, but the conditional approximation here uses \eqref{eq:Hpost}, consistent with Laplace posterior approximation \cite{tierney1986laplace}. We define the partition
\begin{equation}
\mathbf{H}_{\rm post}
=
\begin{bmatrix}
\mathbf{H}_{\nu\nu} & \mathbf{H}_{\nu\tau}\\
\mathbf{H}_{\tau\nu} & \mathbf{H}_{\tau\tau}
\end{bmatrix}.
\label{eq:Hpost_partition}
\end{equation}
When $\mathbf{H}_{\nu\nu}\succ0$ and the Schur complement is positive definite,
\begin{equation}
\mathbf{H}_{\tau\mid\nu}
\triangleq
\mathbf{H}_{\tau\tau}
-\mathbf{H}_{\tau\nu}\mathbf{H}_{\nu\nu}^{-1}\mathbf{H}_{\nu\tau},
\qquad
\boldsymbol{\Sigma}_{\tau}(\mathbf{Y}_E,V_{\mathcal{K}_E})
\approx
\mathbf{H}_{\tau\mid\nu}^{-1}.
\label{eq:H_tau_given_nu}
\end{equation}
Numerical implementations may replace $\mathbf{H}_{\nu\nu}$ by $\mathbf{H}_{\nu\nu}+\epsilon\mathbf{I}$ when ill conditioned, but singular or multimodal cases are not rigorously covered by this local approximation. \textcolor{red}{The conditional Gaussian posterior is given as follows \cite{tierney1986laplace}:}
\begin{equation}
q(\boldsymbol{\tau}|\mathbf{Y}_E,V_{\mathcal{K}_E},\boldsymbol{\omega})
=
\mathcal{N}\big(
\mathbf{m}_{\tau}(\mathbf{Y}_E,V_{\mathcal{K}_E}),
\boldsymbol{\Sigma}_{\tau}(\mathbf{Y}_E,V_{\mathcal{K}_E})
\big),
\label{eq:q_tau}
\end{equation}
where $\mathbf{m}_{\tau}$ is the target subvector of $\widehat{\boldsymbol{\xi}}_{\rm MAP}$. By removing all unknown physical variables $\boldsymbol{\nu}$ at once, this calculation keeps the posterior covariance between the sensing variable $\theta_L$ and the real and imaginary components of the symbols, collected in $\mathbf{s}_{\rm R}$.

\subsection{Conditional OR Success Probability}
For fixed $(\mathbf{Y}_E,V_{\mathcal{K}_E},\boldsymbol{\omega})$, define
\begin{equation}
\begin{aligned}
\widetilde v(\mathbf{Y}_E,V_{\mathcal{K}_E},\boldsymbol{\omega})
\triangleq
\sup_{\widehat{\theta}_L,\widehat{\mathbf{s}}_{\rm R}}
\Pr_q\big(
&|\theta_L-\widehat{\theta}_L|\le\delta_L\\
&\text{or }
\|\mathbf{s}_{\rm R}-\widehat{\mathbf{s}}_{\rm R}\|_2^2
\le n_C\delta_c
\big).
\end{aligned}
\label{eq:conditional_vulnerability}
\end{equation}
\textcolor{red}{where $\Pr_q(\cdot)$ denotes probability with respect to the Gaussian posterior in \eqref{eq:q_tau}.} For any distribution $r$ of $\boldsymbol{\tau}$, $\widetilde{\mathcal V}_{\rm OR}(r)$ is the highest OR success probability calculated from $r$. Thus, $\widetilde v(\mathbf{Y}_E,V_{\mathcal{K}_E},\boldsymbol{\omega})=\widetilde{\mathcal V}_{\rm OR}\big(q(\boldsymbol{\tau}\mid\mathbf{Y}_E,V_{\mathcal{K}_E},\boldsymbol{\omega})\big)$.

The angular component in \eqref{eq:conditional_vulnerability} is treated in a local Euclidean coordinate. Thus, the Laplace approximation is appropriate when the posterior is sufficiently concentrated so that the periodicity of the angle is negligible \cite{tierney1986laplace}. For each observation, Eve chooses its symbol and sensing estimates to maximize the probability of success in either task under the Gaussian posterior. It can be approximated by sample average optimization over posterior samples \cite{shapiro2003montecarlo}. \textcolor{red}{Using the posterior mean or MAP estimates yields a lower complexity approximation $\widetilde P_Y^{\rm plug}$, but these estimates need not maximize the OR success probability under the Gaussian approximation.} The observation-averaged approximation is then defined as:
\begin{equation}
\widetilde{P}_{Y}(\zeta)
\triangleq
\E_{\mathbf{Y}_E,V_{\mathcal{K}_E}}
\big[\widetilde v(\mathbf{Y}_E,V_{\mathcal{K}_E},\boldsymbol{\omega})\big],
\label{eq:PY_tilde_avg}
\end{equation}

\begin{figure}[!b]
    \centering
    \subfloat[\label{fig:snr_compromise}]{\includegraphics[width=0.485\columnwidth]{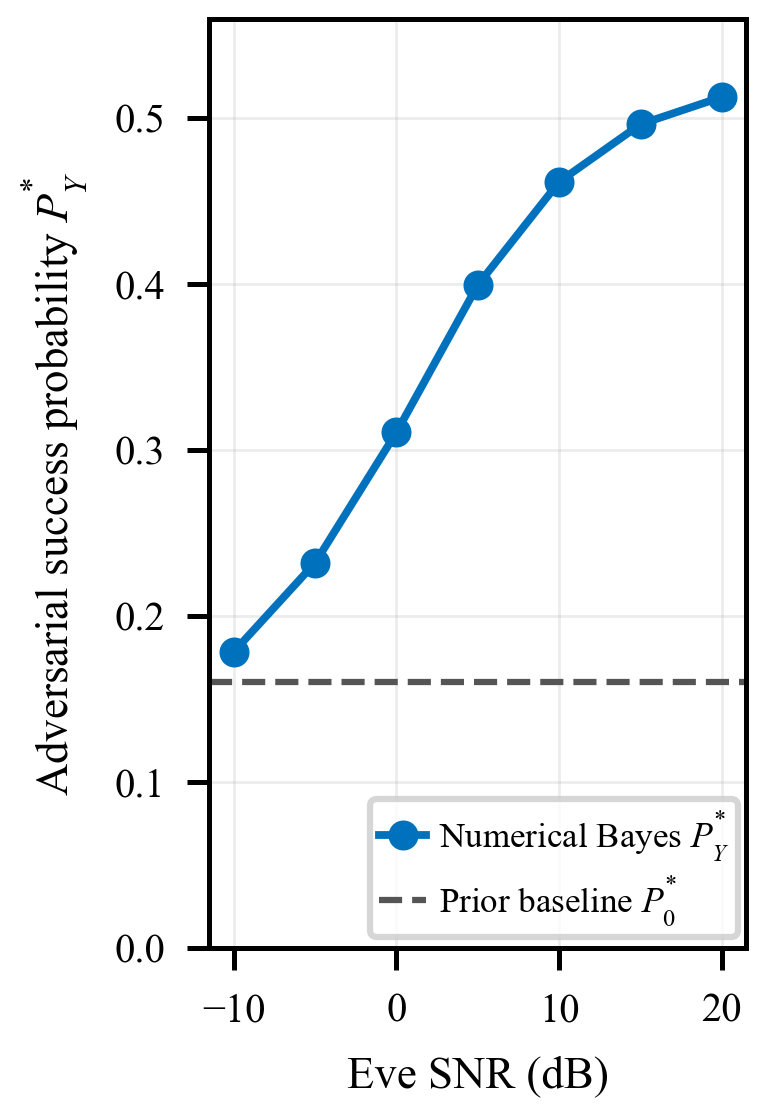}}
    \hfil
    \subfloat[\label{fig:snr_normalized}]{\includegraphics[width=0.485\columnwidth]{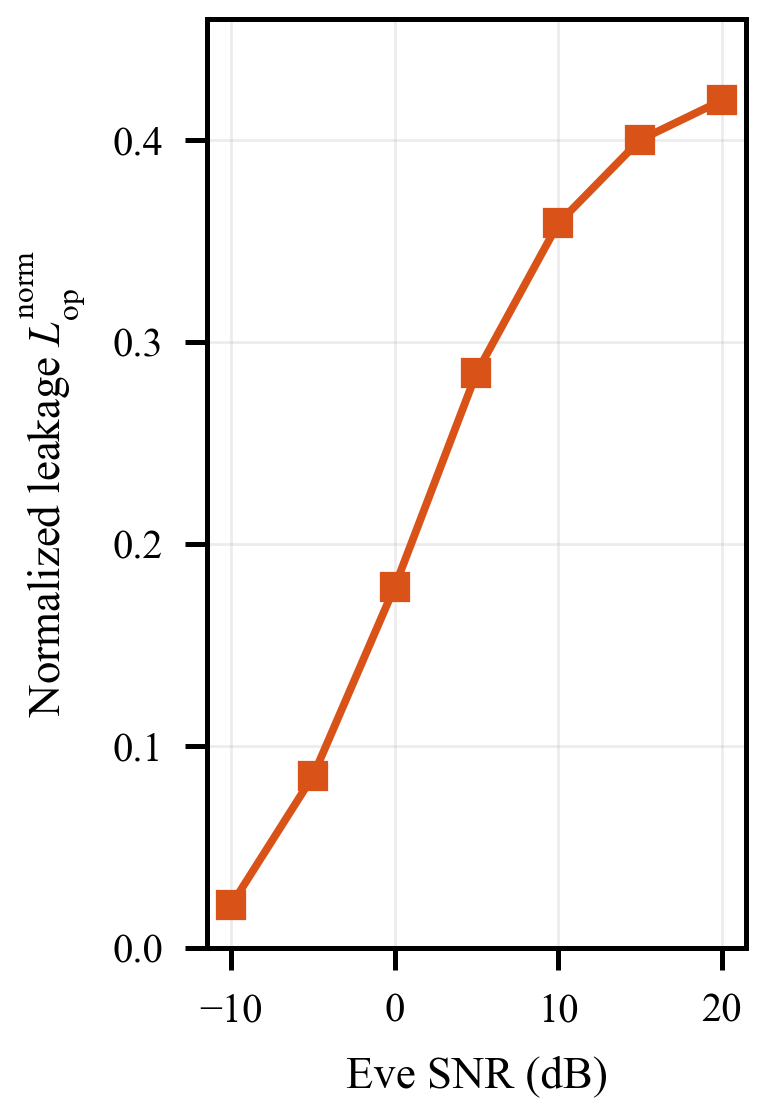}}
    \caption{Observation-assisted adversarial success probability versus Eve SNR at the nominal distortion thresholds. (a) Bayes adversarial success probability $P_Y^\star$ and prior baseline $P_0^\star$. (b) Corresponding normalized observation-induced leakage $L_{\mathrm{op}}^{\mathrm{norm}}$.}
    \label{fig:snr_leakage}
\end{figure}

\begin{figure*}[!t]
    \centering
    \subfloat[\label{fig:secrecy_vs_operational}]{\includegraphics[width=0.32\textwidth]{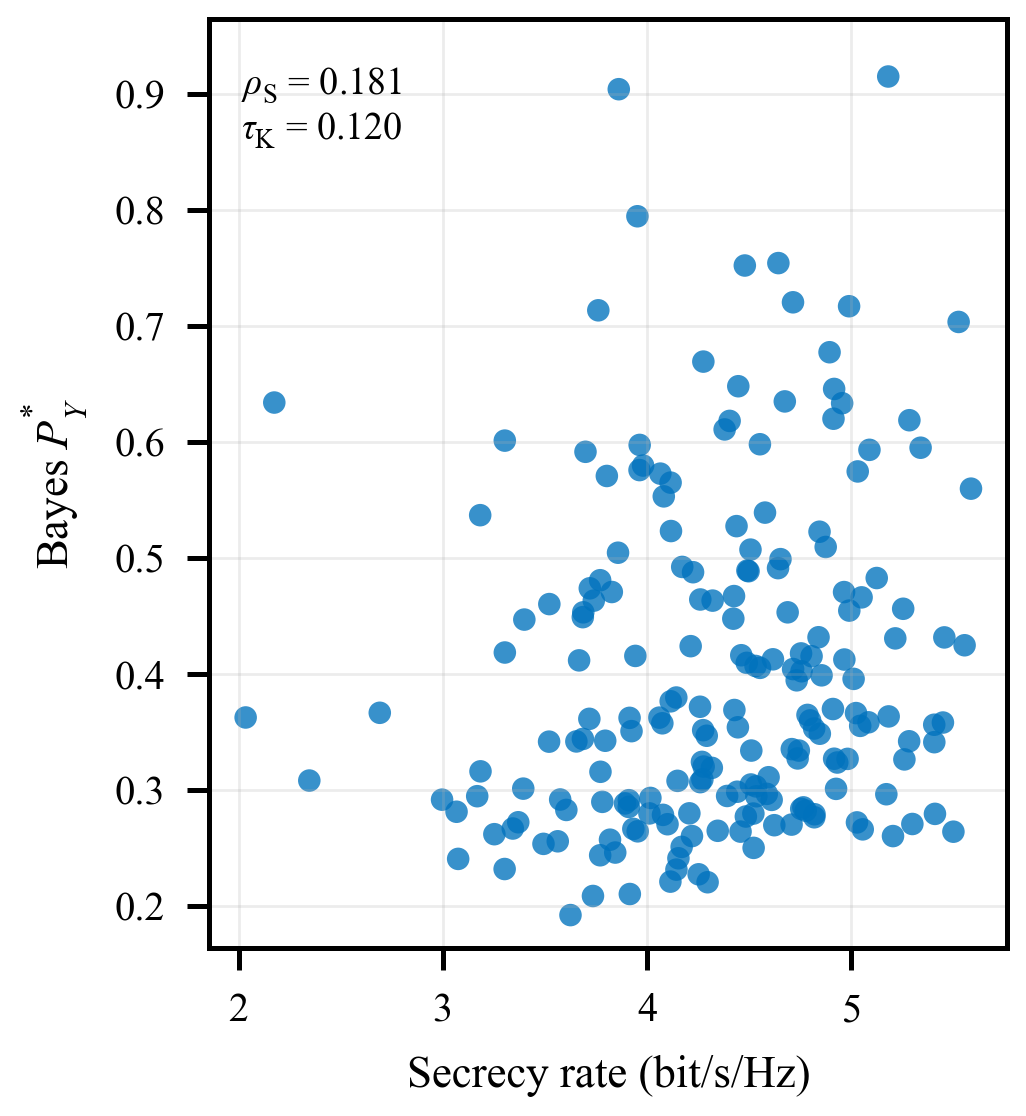}}
    \hfil
    \subfloat[\label{fig:crb_vs_operational}]{\includegraphics[width=0.32\textwidth]{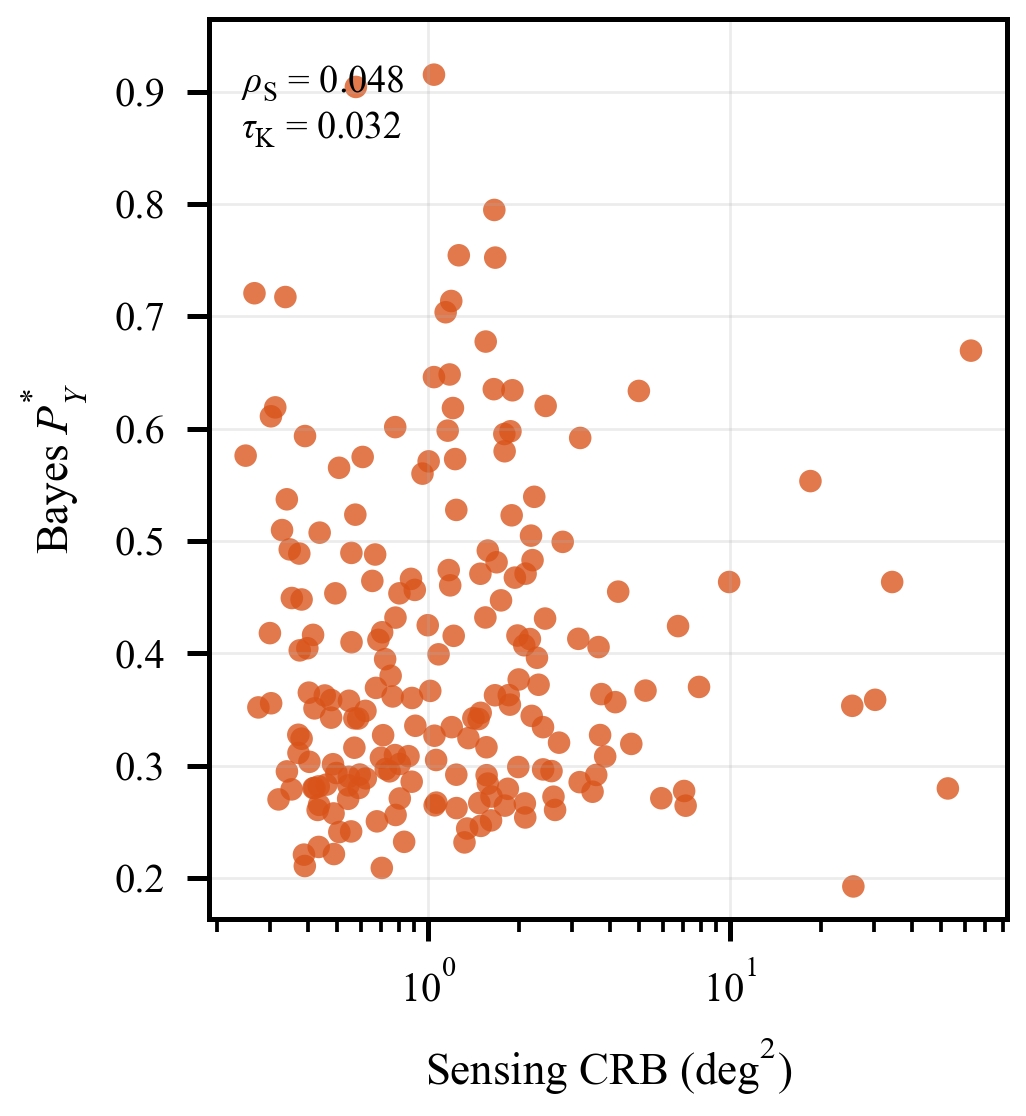}}
    \hfil
    \subfloat[\label{fig:proxy_vs_operational}]{\includegraphics[width=0.32\textwidth]{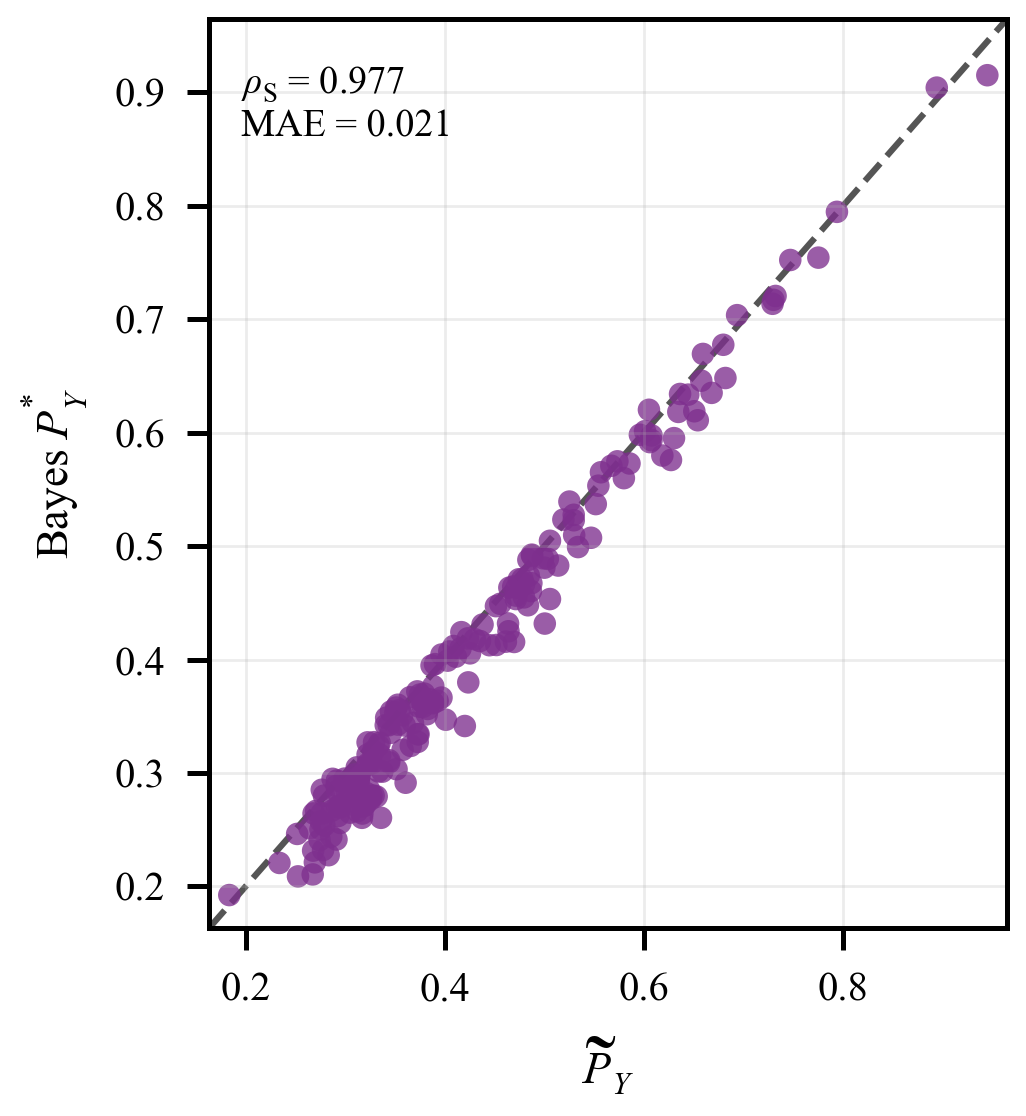}}
    \caption{Comparison between conventional performance surrogates and the proposed adversarial success metric over feasible candidate designs. (a) Secrecy rate versus numerically evaluated Bayes adversarial success probability $P_Y^\star$. (b) Sensing CRB versus $P_Y^\star$. (c) Conditional Gaussian approximation $\widetilde P_Y$ versus Bayes evaluation, where the dashed line denotes $y=x$. Panels (a)--(b) report rank correlations $(\rho_{\mathrm S},\tau_{\mathrm K})$, whereas panel (c) reports $\rho_{\mathrm S}$ and MAE.}
    \label{fig:metric_validation}
\end{figure*}

\begin{figure}[!b]
    \centering
    \subfloat[\label{fig:threshold_comm}]{\includegraphics[width=0.485\columnwidth]{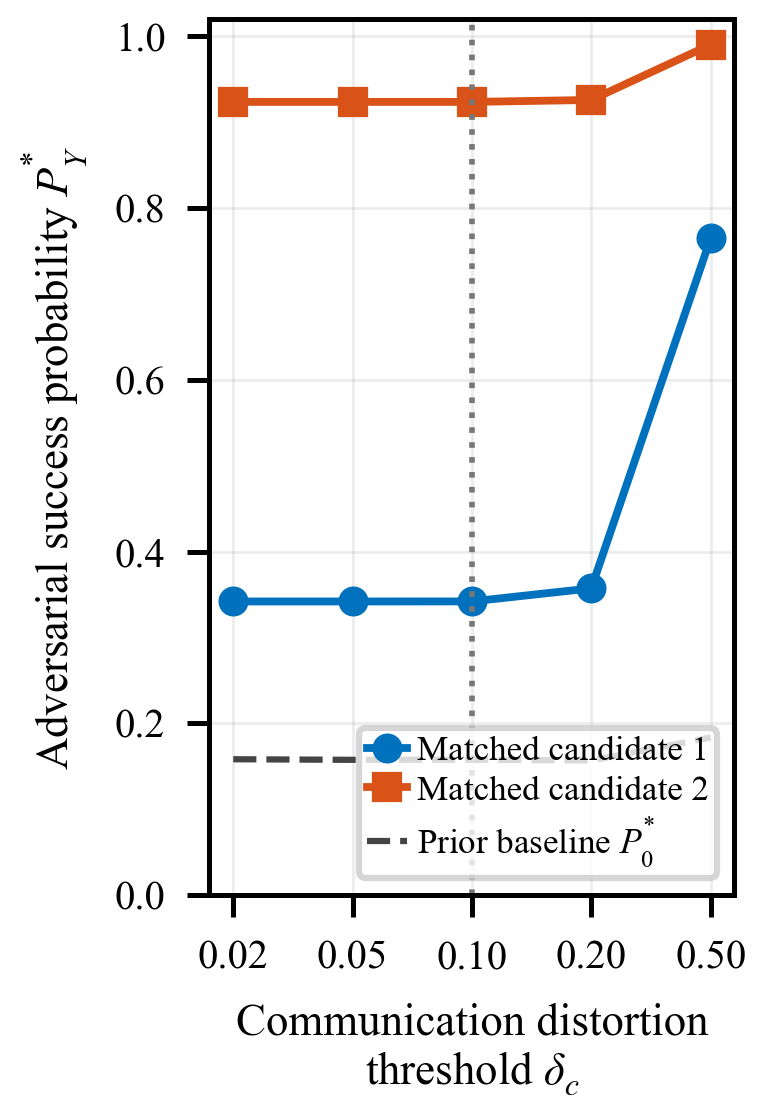}}
    \hfil
    \subfloat[\label{fig:threshold_angle}]{\includegraphics[width=0.485\columnwidth]{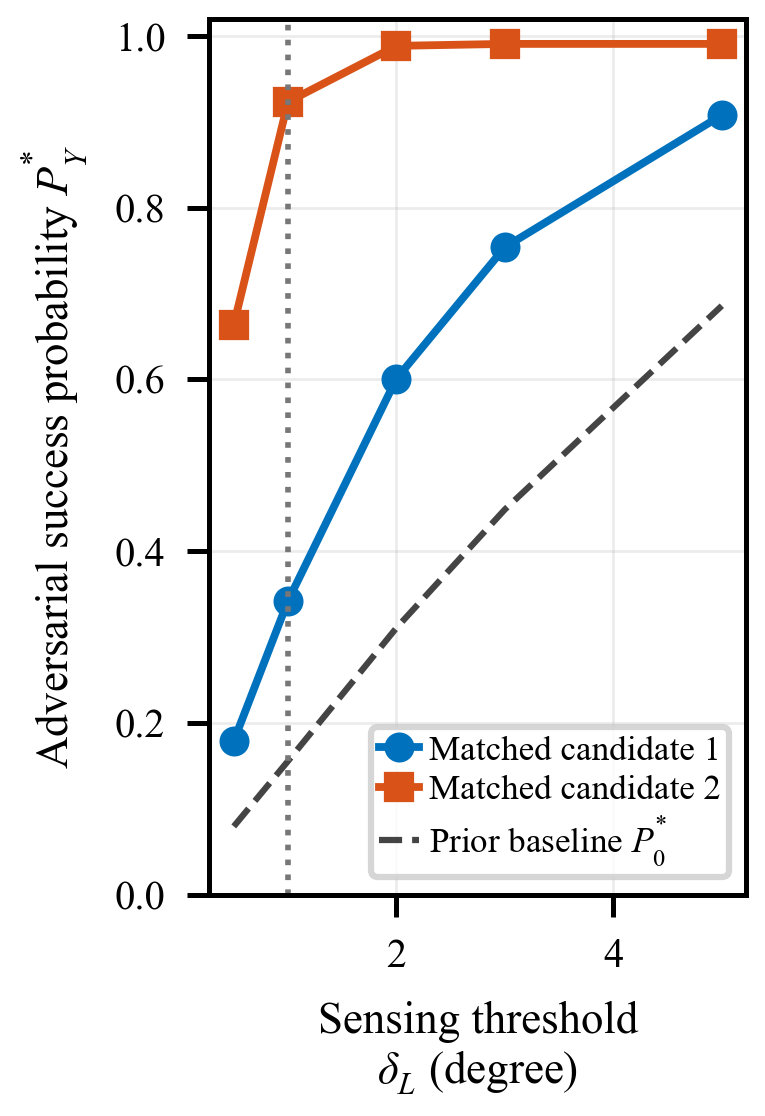}}
    \caption{Threshold dependence of adversarial success probability for two candidates with similar conventional metrics. (a) Communication threshold $\delta_c$ with $\delta_L=1^\circ$. (b) Sensing threshold $\delta_L$ with $\delta_c=0.10$. The dashed curves denote prior adversarial success probability $P_0^\star$, and vertical dotted lines indicate nominal thresholds.}
    \label{fig:threshold_dependence}
\end{figure} 
\noindent which follows the same derivation as \eqref{eq:posterior_characterization}: for each realization of $(\mathbf{Y}_E,V_{\mathcal{K}_E})$, it first maximizes Eve's conditional OR success probability over the symbol and sensing estimates, and then averages the optimized value over $(\mathbf{Y}_E,V_{\mathcal{K}_E})$.

\textcolor{red}{To obtain the corresponding prior value, we apply the same success probability calculation to the prior distribution $q_0(\boldsymbol{\tau}\mid V_{\mathcal{K}_E},\boldsymbol{\omega})$ of $\boldsymbol{\tau}$:}
\begin{equation}
\begin{aligned}
\widetilde{P}_{0}(\zeta)
\triangleq
\E_{V_{\mathcal{K}_E}}
\left[
\widetilde{\mathcal V}_{\rm OR}\big(
q_0(\boldsymbol{\tau}|V_{\mathcal{K}_E},\boldsymbol{\omega})
\big)
\right].
\end{aligned}
\label{eq:P0_tilde}
\end{equation}
It may be the exact prior distribution, an empirical distribution, or a Laplace approximation. Note that $q_0$ and the conditional distributions $q$ may be approximated separately, they may not satisfy $\widetilde P_Y\ge\widetilde P_0$. When $0\le\widetilde P_0\le\widetilde P_Y\le1$, one may also consider the following diagnostic normalized gain
\begin{equation}
\widetilde L_{\rm diag}(\zeta)
\triangleq
\frac{\widetilde P_Y(\zeta)-\widetilde P_0(\zeta)}
{1-\widetilde P_0(\zeta)},\quad \widetilde P_0<1 .
\label{eq:Ldiag_tilde_norm}
\end{equation}
\textcolor{red}{This is a diagnostic approximation of the exact normalized leakage in \eqref{eq:Lop_norm}.}
If the ordering is violated, the pair $(\widetilde P_0,\widetilde P_Y)$ should be reported without normalization.

For fixed symbol and sensing estimates under the Gaussian distribution, the approximated OR success probability is
\begin{equation}
\widetilde P_{\rm OR}=\widetilde P_{\mathrm{s}}+\widetilde P_{\mathrm{c}}-\widetilde P_{\mathrm{s}\cap\mathrm{c}},
\label{eq:OR_decomposition}
\end{equation}
where $\widetilde P_{\mathrm{s}}$, $\widetilde P_{\mathrm{c}}$, and $\widetilde P_{\mathrm{s}\cap\mathrm{c}}$ are the sensing-success probability, communication-success probability, and probability that both events occur for those estimates. The last term must use the joint covariance or joint samples, and not the product of the two individual probabilities. If the estimates equal the corresponding posterior means, define $e_L\triangleq\theta_L-\widehat{\theta}_L$. For a scalar sensing distribution with variance $\sigma_L^2$, the following holds true:
\begin{subequations}\label{eq:task_proxies}
\begin{equation}
\widetilde{P}_{\mathrm{s}}(\delta_L)
\triangleq
\Pr(|e_L|\le\delta_L)
=
1-2Q\!\left(\frac{\delta_L}{\sqrt{\sigma_L^2}}\right).
\label{eq:Q_sense}
\end{equation}
Here $Q(x)=\int_x^\infty (2\pi)^{-1/2}e^{-t^2/2}dt$. If the symbol estimate equals the posterior mean, then $\|\widehat{\mathbf{s}}-\mathbf{s}\|_2^2=\|\widehat{\mathbf{s}}_{\rm R}-\mathbf{s}_{\rm R}\|_2^2$. Let $\boldsymbol{\Sigma}_{s}$ be the block of $\boldsymbol{\Sigma}_{\tau}$ corresponding to $\mathbf{s}_{\rm R}$, let $\lambda_i$ be its eigenvalues, and let $Z_i\sim\mathcal{N}(0,1)$ be independent and identically distributed; then
\begin{equation}
\widetilde{P}_{\mathrm{c}}(\delta_c)
\triangleq
\Pr\!\left(\sum_{i=1}^{2n_C}\lambda_i Z_i^2
\le n_C\delta_c\right),
\label{eq:gchi_comm}
\end{equation}
which is the cumulative distribution function (CDF) of a generalized chi-square random variable. In the isotropic special case
$\boldsymbol{\Sigma}_{s}=\sigma_s^2\mathbf{I}_{2n_C}$, \textcolor{red}{\eqref{eq:gchi_comm}} has the following closed-form expression
\begin{equation}
\widetilde{P}_{\mathrm{c}}(\delta_c)
=
F_{\chi^2_{2n_C}}\!\left(\frac{n_C\delta_c}{\sigma_s^2}\right).
\label{eq:chi_comm}
\end{equation}
\textcolor{red}{Here, $F_{\chi^2_{2n_C}}$ denotes the cumulative distribution function of a chi-square random variable with $2n_C$ degrees of freedom.}
\end{subequations}
The union bound $\min\{1,\widetilde P_{\mathrm{s}}+\widetilde P_{\mathrm{c}}\}$ is only a rough sanity check under the Gaussian approximation and is not a bound on $P_Y^\star$.

\section{Numerical Results}
We considered a BS with $N_{\mathrm{BS},t}=8$ transmit antennas, Eve with $N_E=4$ receive antennas, and one single-antenna downlink user. The intended-user direction, the prior mean of the target direction $\theta_L$, and the Eve-side echo angle $\theta_E$ were $-20^\circ$, $10^\circ$, and $30^\circ$, respectively. We used $\theta_L\sim\mathcal{N}(10^\circ,(5^\circ)^2)$ and, for the single user, $s(\ell)\sim\CN(0,1)$. Unless otherwise specified, we considered $L=100$ snapshots with nominal distortion thresholds $\delta_c=0.10$ and $\delta_L=1^\circ$. For metric validation, we randomly generated 200 feasible $(\mathbf F,\mathbf Z)$ pairs under the same power, communication, and sensing constraints. These candidates were used to compare different security characterizations rather than to optimize the transmit design.

Figure~\ref{fig:snr_leakage} evaluates the proposed adversarial success metric under different Eve SNR conditions. The prior adversarial success probability is $P_0^\star=0.161$ and remains unchanged with Eve SNR. In contrast, the numerically evaluated $P_Y^\star$ increases from $0.179$ at $-10$ dB to $0.513$ at $20$ dB. The normalized leakage in Figure~\ref{fig:snr_leakage}(b) correspondingly increases from $0.021$ to $0.420$, showing that the received observation produces a growing increment in Eve's success probability beyond the prior baseline. Thus, $L_{\mathrm{op}}^{\mathrm{norm}}$ separates the risk due to Eve's side information from the additional risk due to its observation.

To assess whether conventional metrics preserve the ordering of candidate designs by adversarial success probability, we use Spearman's rank correlation coefficient $\rho_{\mathrm S}$ and Kendall's rank correlation coefficient $\tau_{\mathrm K}$ \cite{spearman1904association,kendall1945ties}. Figure~\ref{fig:metric_validation} compares the conventional surrogates with $P_Y^\star$. As observed, the SR surrogate has weak ranking consistency, with $(\rho_{\mathrm S},\tau_{\mathrm K})=(0.181,0.120)$, and the sensing CRB is weaker still, with $(\rho_{\mathrm S},\tau_{\mathrm K})=(0.048,0.032)$. Thus, rate- and estimation-based quantities cannot uniquely determine the finite-record adversarial success probability. \textcolor{red}{In contrast, the conditional Gaussian approximation $\widetilde P_Y$ preserves the ordering of candidate designs and closely follows the numerical Bayesian evaluation, achieving $\rho_{\mathrm S}=0.977$. Its absolute accuracy is quantified by the mean absolute error (MAE):}
\begin{equation}
\mathrm{MAE}\triangleq\frac{1}{N}\sum_{i=1}^{N}\left|\widetilde P_Y^{(i)}-P_{Y,i}^\star\right|,
\label{eq:mae_proxy}
\end{equation}
\textcolor{red}{For the $N=200$ feasible candidate designs, this MAE equals $0.021$, meaning that the average probability error is $0.021$. Therefore, $\widetilde P_Y$ provides an effective low-complexity approximation to the adversarial success probability $P_Y^\star$.}

Figure~\ref{fig:threshold_dependence} shows how the distortion thresholds affect Eve's success probability. Increasing $\delta_c$ or $\delta_L$ enlarges the success region and raises $P_Y^\star$. The two candidates have similar secrecy rates and sensing CRBs, yet their adversarial success probabilities differ across the considered thresholds. \revision{These results motivate the joint design of the precoder $\mathbf{F}$ and AN covariance $\mathbf{Z}$ to minimize operational leakage at thresholds set by the application, subject to power, communication, and sensing requirements.}

\section{Conclusions}
In this letter, ISAC leakage trough Eve’s Bayes probability of meeting either the communication or sensing distortion threshold was quantified. The prior baseline captures the risk due to Eve’s available  knowledge before observation, while the observation-assisted probability captures the additional risk created by the received signal. The proposed conditional Gaussian approximation enables efficient evaluation of the adversarial success probability when the posterior is sufficiently concentrated; numerical Bayes evaluation remains necessary for broad or multimodal posteriors. The numerical results showcased that measures based on the secrecy rate and the CRB do not uniquely determine finite-record adversarial success probability, and that the separation of prior and observation-induced
risk depends on the chosen distortion thresholds.

\end{document}